\documentclass[submission,copyright,creativecommons]{eptcs}
\providecommand{\event}{SOS 2007} 
\usepackage{breakurl}             
\usepackage{underscore}           

\usepackage{mathtools, amsmath, amssymb, tikz, amsthm,
tasks,stackengine, graphicx,enumitem, MnSymbol, float,stmaryrd,mathrsfs} 
\usetikzlibrary{arrows,automata}
\usetikzlibrary{positioning}
\usepackage{slashed}
\usepackage{trimclip}
\usepackage{braket}
\usepackage{amsthm}
\newtheorem{definition}{Definition}
\newtheorem{example}{Example}

\newtheorem{lemma}{Lemma}

\newlist{mylist}{enumerate*}{1}
\setlist[mylist]{label=(\roman*)}

\def\@clipped@vdash{%
  \raise .6ex\hbox{\clipbox{0pt .6ex 0pt .6ex}{$\vdash$}}%
}
\newcommand*\vDdashA{%
  \mathrel{%
    \ooalign{%
      $\vdash$\cr
      \raise  .3ex\hbox{\@clipped@vdash}\cr
      \raise -.3ex\hbox{\@clipped@vdash}%
    }%
  }%
}
\newcommand{\logica}[1]{\mathrm{L}(#1)}
\newcommand{\simnot}{\mathord{\sim}}
\newcommand\mydiamonplus{\ensurestackMath{%
  \stackengine{.5pt}{\Diamond}{\scalebox{.75}[1]{$+$}}{O}{c}{F}{F}{L}}}
  
\newcommand\mydiamonminus{\ensurestackMath{%
  \stackengine{.5pt}{\Diamond}{\scalebox{.75}[1]{$-$}}{O}{c}{F}{F}{L}}}

\newcommand{\A}{\mathbf A}
 \newcommand{\lat}{\langle A,\sqcap,\sqcup, 1,0,\rightharpoonup \rangle}

\newcommand{\toh}{\rightharpoonup}
\newcommand{\Max}{\bigsqcup}
\newcommand{\Min}{\bigsqcap}
\newcommand{\nneg}{\sslash}
\newcommand{\G}{\textbf{\"G}}
\newcommand{\2}{\mathbf 2}

\newcommand{\Prop}{\textrm{Prop}}
\newcommand{\PTS}{\textrm{PTS}}

\newcommand{\Fm}{\textrm{Fm}}

\newenvironment{calculation}{\begin{eqnarray*}&&}{\end{eqnarray*}}

\newcommand{\just}[2]{\\ &#1& \rule{2em}{0pt} \{ \mbox{\rule[-.7em]{0pt}{1.8em} \footnotesize #2 \/} \} \nonumber\\ && }

\allowdisplaybreaks[3]

\title{A Logic for Paraconsistent Transition Systems}
\author{Ana Cruz
\institute{INESC TEC, University of Minho, Portugal}
\and
Alexandre Madeira 
\institute{CIDMA, University of Aveiro, Portugal}
\and
Lu{\'\i}s Soares Barbosa
\institute{INESC TEC, University of Minho, Portugal}
}
\def\titlerunning{A Logic for Paraconsistent Transition Systems}
\def\authorrunning{A. Cruz, A. Madeira \& L.\,S. Barbosa}
\begin{document}
\maketitle

\begin{abstract}
Modelling complex information systems often entails the need for dealing with scenarios of inconsistency in which several requirements either reinforce or contradict each other. In this kind of scenarios, arising e.g. in knowledge representation,  simulation of biological systems, or 
quantum computation, inconsistency has to be addressed in a precise and controlled way. This paper generalises  
Belnap-Dunn four-valued logic,
introducing paraconsistent transition systems (PTS), endowed with positive and negative accessibility relations, and a metric space over the lattice of truth values,  and their modal logic.
\end{abstract}

\section{Introduction}\label{scintrod}

\paragraph{Context and motivation.}

Different sort of many-valued logics have been proposed  to reason about modelling contexts in which the classical bivalent distinction is not enough, in particular to capture vagueness or uncertainty. Residuated lattices, adding a commutative monoidal structure to a complete lattice such that the monoid composition has a right adjoint, the residue,  provide the  semantic universe for such logics. A suitable choice of the lattice carrier, which stands for the set of truth values, does the job --- a typical example being the real $[0,1]$ interval. Reference \cite{Bou} explores, in a systematic way, the modal extensions of many-valued logics whose Kripke frames are defined over (variants of) residuated lattices.

In a number of situations, however, this is not yet the whole picture. Often, indeed, there is also a need to deal simultaneously with what could be called \emph{positive} and \emph{negative} accessibility relations, one weighting the possibility of a transition to be present, the other weighting the possibility of being absent. In most real situations such weights are not complementary, and thus both relations should be incorporated in the Kripke frame. 
Suppose, for illustrative purposes, that weights for both transitions come from a residuated lattice over the real $[0,1]$ interval.
Then, the two accessibility relations jointly express a scenario of 
\begin{itemize}[noitemsep]
\item \emph{inconsistency}, when the positive and negative weights  are contradictory, i.e. they sum to some value greater than $1$
(cf, the upper triangle in Fig.~\ref{fig1} filled in grey).
\item \emph{vagueness}, when the sum is less than $1$ (cf, the lower, periwinkle  triangle in Fig.~\ref{fig1});
\item \emph{strict consistency}, when the sum is exactly $1$, which means that the measures of the factors enforcing or preventing a transition are complementary, corresponding to the red line in the figure.
\end{itemize}

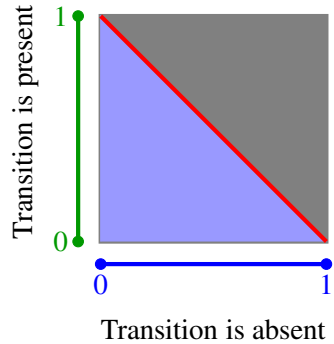
\begin{figure}[H]
\begin{center}
\begin{tikzpicture}
\draw[gray, ultra thick] (0,0) rectangle (3,3);

\draw[black!40!green, ultra thick] (-0.3,0) -- (-0.3,3) ;
\draw[blue, ultra thick] (0,-0.3) -- (3, -0.3);
\draw[] (-1,3.3) node[anchor=east, rotate=90] {Transition is present};
\draw[] (1.5,-0.9) node[anchor=north] {Transition is absent};
\filldraw[blue] (0,-.3) circle (2pt) node[anchor=north] {$0$};
\filldraw[blue] (3,-0.3) circle (2pt) node[anchor=north] {$1$};
\filldraw[black!40!green] (-0.3,0) circle (2pt) node[anchor=east] {$0$};
\filldraw[black!40!green] (-0.3,3) circle (2pt) node[anchor=east] {$1$};
\filldraw[gray] (3,0) -- (3,3) -- (0,3);
\filldraw[white!60!blue] (0,0) -- (0,3) -- (3,0);
\draw[red, ultra thick] (0,3) -- (3,0);
\end{tikzpicture}
\end{center}
\caption{The vagueness-inconsistency square.}
\label{fig1}
\end{figure}

Exploring the three regions in the picture motivated an incursion into paraconsistent logic \cite{js123,Carnielli2016}, whose purpose is to
regard inconsistent information as potentially informative. Such logics were 
originally developed in the decades of 1950 and 1960, mainly by F. Asenjo and Newton da Costa. Quickly, however,
the topic attracted attention in the international community  and the original scope of mathematical applications broadened  out, as witnessed in a recent book  emphasizing the engineering potential of paraconsistency \cite{Aka16}. 

A popular example of the sort of scenarios we want to capture is that of a robot with a sensor capturing the presence of a specific stimulus and another collecting evidence in the opposite directions.
Our own motivation for this work comes, however, from a different context, that of  quantum computation in the context of  NISQ (\emph{Noisy Intermediate-Scale Quantum}) technology \cite{Preskill18} in which levels of decoherence of quantum memory need to be articulated with the length of the circuits to assess program quality. The actual possibility of a transition, approached simultaneously, from the affirmative and the negative perspectives, i.e. the possibility of the system remaining coherent during the execution of a program step, or becoming unstable, respectively, provides a basis to analyse the impact of qubit decoherence in quantum circuits optimization. 

This example is discussed in the recent MSc dissertation of the first author \cite{Ana21}. This paper, however, focuses on the characterization of the relevant transition structures, which we call \emph{paraconsistent transition systems}, abbreviated to PTS in the sequel, and the modal logic for which they provide a semantics. An important remark is  that all constructions discussed in the sequel are parametric in a class of residuated lattices, thus admitting different instances according to the structure of the truth values domain that better suits each modelling context.  
Moreover, the semantics is endowed with a metric over the lattice carrier to characterise the `grey-periwinkle square´  above, in a general setting. 

\paragraph{A tribute to J. Dunn.}

In a sense this paper aims at generalising the Belnap-Dunn four-valued logic \cite{10.1093}, allowing transition weights to be taken from a larger class of residuated lattices. Recall that in the latter propositions are interpreted in a bilattice over  $\{T,F\}$, typically denoted by $\mathcal{FOUR}$. Bilattice elements are pairs $(a,b)$ such that $a, b\in \{T,F\}$. Thus, the valuation of propositions assigns one out of four possible truth values: $\{(T,F),(F,T),(T,T),(F,F)\}$, understood as representing not only the classical Boolean values, but \emph{excess} of information (i.e. inconsistency, $(T,T)$), and  its \emph{insufficiency} (i.e. vagueness, $(F,F)$), respectively. 
The bilattice $\mathcal{FOUR}$ expresses  both a 
 truth and an information ordering, $\leq_t$ and $\leq_i$, each forming a complete lattice as depicted below. 
 Formally, 

$$(a,b)\leq_t (c,d) \text{ iff } a\leq c \text{ and } b\geq d \; \; \; \text{and}\; \; \;  (a,b)\leq_i (c,d) \text{ iff } a\leq c \text{ and } b\leq d 
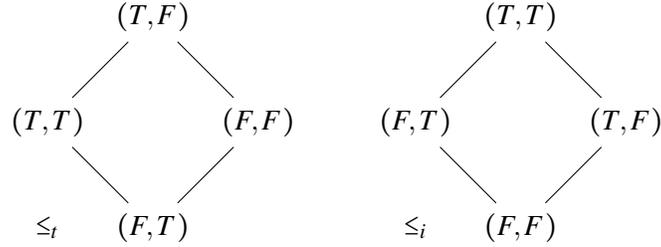
$$


\begin{center}
\begin{tikzpicture}[scale=.7]
  \node (one) at (0,2) {$(T,F)$};
  \node (a) at (-2,0) {$(T,T)$};
  \node (d) at (2,0) {$(F,F)$};
  \node (zero) at (0,-2) {$(F,T)$};
  \node (x) at (-2,-2) {$\leq_t$};
  \draw (zero) -- (a) -- (one) -- (d) -- (zero);
\end{tikzpicture}
\qquad 
\begin{tikzpicture}[scale=.7]
  \node (a) at (2,0) {$(T,F)$};
  \node (one) at (0,2) {$(T,T)$};
  \node (zero) at (0,-2) {$(F,F)$};
  \node (d) at (-2,0) {$(F,T)$};
  \node (x) at (-2,-2) {$\leq_i$};
  \draw (zero) -- (a) -- (one) -- (d) -- (zero);
\end{tikzpicture}
\end{center}


In the modal version of Belnap-Dunn logic,  Kripke frames are built around a four-valued accessibility relation, incorporating in the transitions the notions of inconsistency and vagueness, already considered at the propositional level. 

This paper follows a similar path,  in a more general setting. Actually,  the definition of a bilattice over $\{T,F\}$  does not generalize well to other domains. 
We discuss such a generalization from a different perspective.
Section \ref{scPTS} introduces paraconsistent transition systems (PTS) and reviews the underlying semantic structures. 
The corresponding modal logic is discussed in \ref{sclogic}. Section \ref{secm} characterises bisimilarity for PTS and proves an invariance result.

\section{$\A$-paraconsistent transition structures} \label{scPTS}

As mentioned above, a paraconsistent transition system will be defined over a set of states through a \emph{positive} and a \emph{negative} accessibility relation, taking weights from a common universe. Equivalently, a transition between two states can be regarded as weighted by a pair of values. In a similar way, the valuation of a proposition in a state will be a pair of values capturing, respectively, the degree upon which it may be considered to hold, and to fail. Note that pairs of weights labelling a transition or evaluating a proposition are not necessarily complementary.

This entails the need for making precise i) where weights come from; ii) how is a 'distance' between two such weights computed, for example to assess the level of vagueness or inconsistency in a given transition, and, finally, iii) how pairs of weights are composed. The following subsections address these issues. Then, the definition of a paraconsistent transition system emerges as expected.  

\subsection{Weights}

As usual in multi-valued logics, residuated lattices provide the underlying semantic structure. 
 
A residuated lattice $\langle  A, \sqcap,\sqcup,1,0,\odot ,\rightharpoonup, e\rangle$ over a set $A$ is a complete lattice
$\langle A, \sqcap,\sqcup,1,0\rangle$, equipped with a monoid  $\langle A,\odot, e\rangle$  such that $\odot$ has a right adjoint,
$\rightharpoonup$, called the residuum. We will, however, focus on a particular class of residuated lattices in which the lattice meet ($\sqcap$) and monoidal composition ($\odot$) coincide. Thus the adjunction is stated as
$$a\sqcap b\leq c \; \; \; \; \text{iff}\; \; \; \;  b\leq a\rightharpoonup c$$
Additionally, we will enforce a prelinearity condition
 \begin{equation}\label{prelin}
    (a\rightharpoonup b)\sqcup (b\rightharpoonup a) =1
\end{equation}

A residuated lattice obeying prelinearity  is known as a MTL-algebra in the literature\footnote{From \emph{monoidal t-norm} based logic.} \cite{EstevaFrancesc}.
With a slight abuse of nomenclature, the designation iMTL-algebra, from \emph{integral MTL-algebra}, will be used in the sequel  for the class of semantic structures considered, i.e. prelinear, residuated lattices such that $\sqcap$ and $\odot$ coincide.

This structure enjoys a number of properties. Some are direct consequences of the adjunction, namely
\begin{align}
\label{eq25} \text{ if }\; \; a\leq b \text{ then } (c\rightharpoonup a)\leq (c\rightharpoonup b) \\
\label{eq26}\text{ if } a\leq b \text{ then } (b\rightharpoonup c)\leq (a\rightharpoonup c)
 \end{align}
 and
 \begin{equation}
 a \rightharpoonup (b\rightharpoonup c) = (b\sqcap a) \rightharpoonup c
 \end{equation}




The interplay between the residuum and (the distributed versions of) of $\sqcap$ and $\sqcup$ is captured in the following lemma.

\begin{lemma}\label{lmaux} Let ${\mathbf A} = \langle A,\sqcap,\sqcup,1,0,\rightharpoonup\rangle$ be an  iMTL-algebra. Then, for any $a_1,\dots,a_n,b\in A$
 \begin{align} 
 \label{eq1lm1}
  b \rightharpoonup \big(\bigsqcap_{i} a_i\big) = \bigsqcap_{i}( b \rightharpoonup a_i)\\
\label{eq2lm1} \big(\bigsqcup_{i} a_i\big)\rightharpoonup b = \bigsqcap_{i}(a_i \rightharpoonup b) \\
b\rightharpoonup\big (\bigsqcup_i a_i\big ) = \bigsqcup_i \big ( b \rightharpoonup a_i\big ) \label{eq3lm1} \\ \label{eq4lm1}
\big ( \bigsqcap_i a_i \big ) \rightharpoonup b = \bigsqcup_i \big ( a_i \rightharpoonup b \big)\end{align}
\end{lemma}

\begin{proof}
Properties (\ref{eq1lm1}) and (\ref{eq2lm1}) are true for any residuated lattice \cite{Bou}. Equalities  (\ref{eq3lm1}) and (\ref{eq4lm1}) follow from properties \ref{eq25} and \ref{eq26}, respectively, but also require  the prelinearity condition \cite{MADEIRA20161011}.
\end{proof}

 \begin{example} \label{GodelAlgebra} The Boolean algebra over $\{0,1\}$ with the usual Boolean operators, $\textbf{2}=\langle \{0,1\},\wedge,\vee,1,0, \to\rangle$ is, of course, an iMTL-algebra. A more interesting example introduces a third element in the carrier, $u$, as a denotation of \emph{unknown}. The resulting iMTL-algebra $\mathbf{3}=\langle\{\bot, u, \top\},\wedge_3,\vee_3,\top,\bot, \rightarrow_3\rangle$ is the chain of three elements, where
 
		\begin{center}
		\begin{tabular}{l|rrr} 
			$\vee_3$ & $\bot$ & $u $ & $\top$\\
			\hline
			$\bot$ & $\bot$ & $u$ & $\top$ \\
			$u$ & $u$ & $u$& $\top$ \\
			$\top$ & $\top$ & $\top$&$\top$ 
		\end{tabular} \;\;\;
		\begin{tabular}{l|rrr} 
			$\wedge_3$ & $\bot$ & $u$& $\top$\\
			\hline
			$\bot$ & $\bot$ & $\bot$ &$\bot$ \\
			$u$ & $\bot$ & $u $&$u$ \\
			$\top$ & $\bot$ & $u$ &$\top$ 
		\end{tabular}
	\;\; \;
	\begin{tabular}{l|rrr} 
			$\rightarrow_3$ & $\bot$ & $u$ &$\top$\\
			\hline
			$\bot$ & $\top$ & $\top $&$\top$ \\
			$u$    & $\bot$    & $\top$ &$\top$ \\
			$\top$ & $\bot$ & $u$&$ \top$ 
		\end{tabular} 
		\end{center}

\noindent
Finally, a well-known example from fuzzy logic research, is  \emph{G\"odel algebra} 
$\G=\langle [0,1], min, max,0,1,\to\rangle$,
where $max$ and $min$ retain the usual definitions, and implication is given by
 $$
 a \to b = \begin{cases} 1, \text{ if } a \leq b \\ b, \text{ otherwise } \end{cases} $$
\end{example}

 \subsection{Adding a metric}
 
%
 
In order to give a concrete meaning to the \emph{vagueness-inconsistency square} depicted in Fig.~\ref{fig1}, one needs to endow  any iMTL-algebra $A$ with a suitable metric, i.e. a function $d:A\times A\to \mathbb{R}^+$ such that
 $d(x,y)=0$ iff $x=y$, and $d(x,y)\leq d(x,z) + d(z,y)$. Thus, we will focus on iMTL-algebras $\mathbf A$ whose carrier $A$ supports a metric space $(A,d)$, for a suitable choice of $d$.


%

	\begin{example}\label{ex.metricLattices}
	Metrics for iMTL-algebras $\mathbf{2}$ and $\mathbf{3}$ are shown below.
\begin{center}
\begin{tabular}{c|c}
$\mathbf{2}$ & $\mathbf{3}$\\
\hline \vspace{0.3cm}
$d(x,y)=\begin{cases} 0 & \text{ if } x=y \\
	1 & \text{ otherwise}\end{cases}$
\; \; 	& \; \; 
\begin{tabular}{l|rrr} 
			$d$ & $\bot$ & $u$ &$\top$\\
			\hline
			$\bot$ & $0$ & $1 $&$2$ \\
			$u$ & $1$ & $0$ &$1$ \\
			$\top$ & $2$ & $1$&$ 0$ 
		\end{tabular}	
\end{tabular}
\end{center}

	\noindent
	On its turn,  $\G$ can be endowed with the Euclidean metric, 
	$d(x,y)=\sqrt{(x-y)^2}$.
	\end{example}

\subsection{An algebra of (pairs of) weights}

If the values of both transitions and propositions are pairs of weights, an algebra is needed to operate upon them. We resort to the notion of a twist-structure \cite{Kracht1998-KRAOEO}, enriched with a metric. Formally,

\begin{definition}\label{twistG} Given a iMTL-algebra $\A=\lat$ enriched with a metric $d:A\times A\to \mathbb{R}^+$, a \emph{metric $\A$-twisted algebra}
$\A^2=\langle A\times A,\doublesqcap, \doublesqcup, \Longrightarrow, \nneg, D \rangle$ is defined as follows: for any  $(a,b)$, $(c,d)$ $\in A \times A$
 \begin{itemize}[noitemsep]
\item $ (a, b) \doublesqcap (c,d) =(a \sqcap c, b \sqcup 
d) $
\item $(a, b) \doublesqcup (c,d) = (a \sqcup c, b \sqcap d) $
\item $(a,b) \Longrightarrow (c, d) = (a \rightharpoonup c, a \sqcap d)$
\item $\nneg (a, b) =(b, a)$
\item $D((a,b),(c,d))=\sqrt{d(a,c)^2+d(b,d)^2}$ is the metric over $A\times A$ induced by $d$.
\end{itemize}
\end{definition}

\noindent
The order in $\A$ is lifted to $\A^2$  as
\[(a,b)\preccurlyeq(c,d)\text{ if and only if }a \leq c \text{ and }b\geq d\]


\medskip


\noindent Under this definition, we are now able to characterize precisely the sets of paraconsistent, consistent and strictly consistent pairs, $\Delta_P$, $\Delta_C$ and $\Delta$, as for Fig.~\ref{fig1}. Thus, in the same order, 
		$$\Delta_P=\{(a,b)\vert D((a,b),(1,1))\leq D((a,b),(0,0)) \}$$
		$$\Delta_C=\{(a,b)\vert D((a,b),(0,0))\leq D((a,b),(1,1)) \}$$
	$$\Delta=\Delta_P \cap \Delta_C$$
	
	

\begin{example}
Let us revisit the iMLT-algebras considered in  Example~\ref{ex.metricLattices}. 
\begin{itemize}
    \item 
As expected, in $\2^2$, the four-element lattice $\mathcal{FOUR}$  mentioned in section \ref{scintrod}, $(1,1)\in \Delta_P$, since $$D((1,1),(1,1))=\sqrt{d(1,1)+d(1,1)}=0\; \; \leq\; \; D((0,0),(1,1))=\sqrt{d(0,1)+d(0,1)}=\sqrt{2}$$
Hence,  $\Delta = \{(1,0),(0,1)\}$, $\Delta_P=\{(1,1)\}$ and $\Delta_C = \{(0,0),(1,0),(0,1)\}$. 


\item
For $\mathbf{3}^2$, observe, for instance, that $(\top,u)\in \Delta_P$, since 
\[D((\top,u),(\bot,\bot))= \sqrt{d(\top,\bot)^2+d(u,\bot)^2}=\sqrt{5} \; \; \text{and}\; \; D((\top,u),(\top,\top))=\sqrt{d(\top,\top)^2+d(u,\top)^2}=1.\] Hence,  
$\Delta_P =\{(\top,\bot),(u,u),(\bot,\top),(\top,u),(u,\top),(\top,\top)\}\, $, $\Delta_C =\{(\top,\bot),(u,u),(\bot,\top),(\bot,\bot),(u,\bot),(\bot,u)\}$ and $\Delta=\{(\top,\bot),(u,u),(\bot,\top)\}$. 

\item 
 Finally, in $\G^2$, 
 observe that, for any $(a,b)\in [0,1]\times [0,1]$,
 \[D((a,b),(1,1))=\sqrt{d(a,1)^2 + d(b,1)^2}=\sqrt{(a^2+b^2)-2(a+b)+2}\] and, if $a+b>1$, $ \sqrt{a^2+b^2}=D((a,b),(0,0))$. Therefore $\Delta_P=\{(a,b)\in [1,0]\times[1,0]\mid a+b>1\}$, $\Delta_C= \{(a,b)\in [1,0]\times[1,0]\mid a+b<=1\}$, and $\Delta=  \{(a,b)\in [1,0]\times[1,0]\mid a+b=1\}$,  as expected. 
\end{itemize}
\end{example}

\subsection{Paraconsistent transition systems}


We are now ready to define paraconsistent transition systems parametric on a iMTL-algebra $\mathbf{A}$. Notice that, in the sequel, given a value $(a,b)$,  $(a,b)^+$  denotes $a$ and $(a,b)^- $  denotes $b$.

\begin{definition}\label{plts} Let $\A$ be an  iMTL-algebra. An \emph{$\A$-paraconsistent transition system} ($\PTS$) is a structure $(W,R)$ where $W$ is a finite (non-empty) set of \emph{states}, and
$ R \subseteq W\times W \times A \times A$
   is a relation over $W$ and the carrier of  $\A$ such that between any pair of states $w_1,w_2 \in W$, there is at most one transition.
\end{definition}

\begin{example}\label{exframes} The definition is illustrated with two examples over  $\G$.

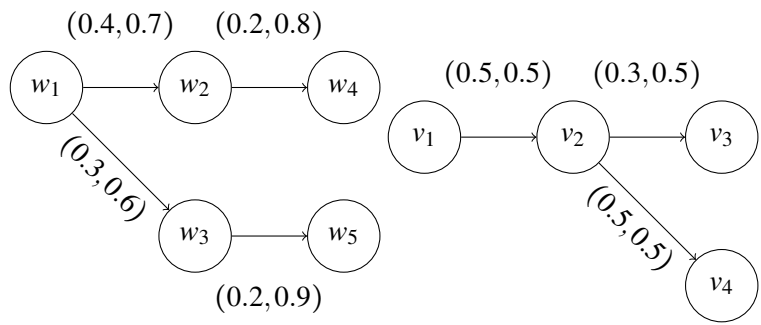
\begin{center}
\begin{tikzpicture}[->]{
   \node[state] (1) {$w_1$};
    \node[state] (2) [right = of 1] {$w_2$};
    \node[state] (3) [below = of 2] {$w_3$};
    \node[state] (4) [right = of 2] {$w_4$};
    \node[state] (5) [right = of 3] {$w_5$};
  \draw (1) edge[above] node[above=0.5] {$(0.4,0.7)$} (2);
    \draw (1) edge[above] node[below,rotate=-45] {$(0.3,0.6)$} (3);
    \draw (2) edge[above] node[above=0.5] {$(0.2,0.8)$} (4);
  \draw (3) edge[above] node[below=0.5] {$(0.2,0.9)$} (5);}
\end{tikzpicture}
\begin{tikzpicture}[->][scale=.6]
    \node[state] (1) {$v_1$};
    \node[state] (2) [right = of 1] {$v_2$};
    \node[state] (3) [right = of 2] {$v_3$};
    \node[state] (4) [below = of 3] {$v_4$};
    \draw (1) edge[above]  node[above=0.5] {$(0.5,0.5)$} (2);
    \draw (2) edge[above] node[above=0.5] {$(0.3,0.5)$} (3);
    \draw (2) edge[above] node[below, rotate=-45] {$(0.5,0.5)$} (4);
\end{tikzpicture}
\end{center}
\end{example}


A tuple $(w_1,w_2,a,b)\in R$ represents  a transition from $w_1$ to $w_2$ which occurs with a certainty degree $a$, and is prevented to occur  with a certainty degree $b$.  Note that a PTS over $\mathbf{2}$ in which all transitions are weighted by $(1,0)$ is just a  standard, deterministic transition system.

The fact that each pair of states can only participate in one transition,
makes it possible to write $R(w_1,w_2)=(a,b)$ for $(w_1,w_2,a,b)\in R$,
as well as to express $R$ in terms of a positive and a negative accessibility relation  $R^+,R^-:W \times W \longrightarrow A$ such that for $w,w'\in W$, 

\[R^+(w,w')=\begin{cases}  a \text{ if }(w,w',a,b)\in R \\ 0 \text{ otherwise } \end{cases}
\; \text{ and }\; \; \;  
R^-(w,w')=\begin{cases}  b \text{ if }(w,w',a,b)\in R \\ 0 \text{ otherwise } \end{cases}\].

Thinking of PTS as Kripke frames, we may define the corresponding notion of a model over an arbitrary iMTL-algebra, as follows:

\begin{definition} Let $\A=\lat$ be an iMTL-algebra and $\Prop$  a set of propositions. 
A \emph{${\A}$-paraconsistent Krikpe model}  is a triple $M=(W,R,V)$ where 
$(W,R)$ is a PTS over $\A$
and $V :\Prop \times W \rightarrow A \times A$ is a \emph{valuation function}. For $p\in\Prop$ and $w\in W$, $V(p,w)=(a,b)$ gives a pair of truth values $(a,b)\in A \times A$, where $a$ (respectively, $b$) is the degree upon which $p$ is considered to hold (respectively, not to hold) in state $w$.

\end{definition}




\section{The logic}\label{sclogic}
This section introduces a logic $\logica{\A}$ to specify paraconsistent Kripke models parametric on a iMTL-algebra $\A$. 

\subsection{Syntax} 
\label{subsyntax}

\begin{definition}[Formulas]
    Given an iMTL-algebra $\A$ and a set of propositions $\Prop$, the formulas $\Fm(\Prop)$ of $\logica{\A}$ are generated by the following grammar:
\[\varphi 
:= p \;\vert\; \bot 
  \;\vert\; \neg\varphi \;\vert\; 
  \varphi \wedge \varphi \vert \varphi \vee \varphi \;\vert\;
  \varphi \rightarrow \varphi \;\vert\;
  \openbox 
  \varphi \;\vert\;
   \Diamond \varphi \;\vert\; \slashed{\openbox} \varphi \;\vert\; \slashed{\Diamond} \varphi 
   \;\vert\; \circ \varphi
\]
where $p\in\Prop$. 
\end{definition}
\noindent As usual, the following abbreviations are assumed:  $\top$ for $\neg\bot$, and $\varphi_1\leftrightarrow\varphi_2$  for $(\varphi_1\rightarrow\varphi_2)\wedge(\varphi_2\rightarrow\varphi_1)$. 

Informally, connectives $\bot$, $\neg$, $\wedge$, $\vee$ and $\to$ stand for false, negation, conjunction, disjunction and implication, respectively. 
There are two sets of modalities over positive and negative information, respectively. Together with negation, modalities are discussed in next sub-section. 
Finally, $\circ$ is the unary consistency connective.


\subsection{Satisfaction relation}
The satisfaction relation in $\logica{\A}$ evaluates the satisfaction of a formula as pairs $(a,b)\in A\times A$ at a given state, in terms of the positive and negative information grades at hands. 
Modalities, in particular, explore both components of the accessibility relation defining the underlying $\A$-paraconsistent Kripke model, through auxiliary operators $\boxplus,\mydiamonplus$
and $\boxminus,\mydiamonminus$.

\begin{definition}[Satisfaction relation] 

Let $\A=\lat$ be an iMTL-algebra and  
$M=(W,R,V)$ an $\A$-paraconsistent Kripke model, where $W$ is a set of states, $R$ is the paraconsistent accessibility relation and $V$ a valuation function. The \emph{satisfaction relation for $M$ in $\logica{\A}$} is a function
 \[\models \;: W\times \Fm(\Prop) \longrightarrow A \times A\] defined for each $\varphi\in\Fm(\Prop)$ as follows:
\begin{itemize}
\item $(w\models p)=V(p,w)$
\item $(w \models \bot) = (0,1)$ 
\item $(w\models \neg \varphi) = \nneg (w\models \varphi)$
\item $(w\models \varphi_1 \wedge \varphi_2)= (w\models\varphi_1) \doublesqcap (w\models\varphi_2)$
\item $(w\models \varphi_1 \vee \varphi_2)= (w\models\varphi_1)\doublesqcup (w\models\varphi_2)$
\item  $(w\models \varphi_1 \to \varphi_2)= (w\models\varphi_1) \Longrightarrow (w\models\varphi_2)$
\item  $(w \models \openbox \varphi)= (\boxplus (w, \varphi^+),  \mydiamonplus (w,\varphi^-))$
\item  $(w \models \Diamond \varphi)= (\mydiamonplus (w,\varphi^+), \boxplus (w,\varphi^-))$
\item  $(w \models \slashed{\openbox} \varphi ) = (\mydiamonminus (w,\varphi^-), \boxminus (w,\varphi^+))$
\item $(w\models \slashed{\Diamond} \varphi)= (\boxminus (w,\varphi^-), \mydiamonminus (w,\varphi^+)) $
\item $(w\models \circ\varphi) = \begin{cases} (1,0) \text{ if }(w\models \varphi) \in \Delta_C \\ (0,1) \text{ otherwise}\end{cases}$
\end{itemize}
The definition resorts to the following operators
\begin{itemize}
 \item 
 $\boxplus (w,\varphi ^*) = \bigsqcap_{w'\in R[w]} (R^+(w,w')\rightharpoonup (w'\models \varphi)^*)$  
 \item $\boxminus (w,\varphi ^* )= \bigsqcap_{w'\in R[w]} (R^-(w,w')\rightharpoonup (w'\models \varphi)^*)$
 \item $\mydiamonplus (w,\varphi ^* )= \bigsqcup_{w'\in R[w]} (R^+(w,w')  \sqcap (w'\models \varphi)^*)$
 \item $\mydiamonminus (w,\varphi ^*) = \bigsqcup_{w'\in R[w]} (R^-(w,w') \sqcap (w'\models \varphi)^*)$
\end{itemize} 

\noindent
where $*\in\{^+,^-\}$ and $R[w]=\{w'|(w,w',a,b)\in R, \text{ for some } a,b \in A \}$.
 Two formulas $\varphi, \psi \in \Fm(\Prop)$ are semantically equivalent, in symbols $\varphi \equiv \psi$, if for any $w\in W$, $(w\models \varphi)=(w\models \psi)$. We say that $\varphi$ is valid if, for any $w\in W$, $(w\models \varphi)=(1,0)$.
\end{definition}

\begin{example}
Consider the $\G$-paraconsistent Kripke model for $\Prop= \{p\}$ defined over the left \G-PTS of Example~\ref{exframes}, with 
$V(w_2,p)=(1,0)$, $V(w_3,p)=(0,1)$ and 
    $V(w,p)=(0,0)$ for $w\in W \setminus \{w_2,w_3\}$. Then, 
{
\begin{align*}
(w_1 \models \openbox p) & =  (\boxplus (w, p^+),  \mydiamonplus (w,p^-))\\
 & = \big(\bigsqcap_{w'\in R[w_1]}(R^+(w_1,w')\rightharpoonup (w'\models p)^+),\bigsqcup_{w'\in R[w_1]}(R^+(w_1,w')\sqcap (w'\models p)^-\big)\\
 & =  \big((R^+(w_1,w_2) \rightharpoonup (w_2\models p)^+)\sqcap (R^+(w_1,w_3) \rightharpoonup (w_3\models p)^+), \\
 &  \;\;\;\; (R^+(w_1,w_2) \wedge (w_2\models p)^-) \sqcup (R^+(w_1,w_3) \sqcap (w_3\models p)^-)\big) 
 \\
 & =  (\min(1,0),  \max(\min(0.4,0),\min(0.3,1))) \\
 & =  (\min(1,0),\max(0,0.3)) =  (0, 0.3)
\end{align*} 
}
\end{example}

A few comments are in order with respect to both modalities and negation. Let us start with the latter.

The satisfaction of a formula $\varphi\in Fm(Prop)$ is given by a pair $(a,b)\in A\times A$ assigning truth degrees $a$ and $b$ to the sentences "$\varphi$ holds" and "$\varphi$ does not hold", respectively. The satisfaction of $\neg\varphi$ is simply the converse $(b,a)$, which has a straightforward interpretation. This negation is an involution and satisfies the so-called pre-minimality condition: $\neg \psi \leq \neg \phi\;  \text{if} \; 
\phi \leq  \psi$. Actually, in the underlying semantics,
 \begin{align*} 
& \nneg (a', b') \preccurlyeq \nneg (a, b) \; \; \text{if} \; \;
    (a, b) \preccurlyeq (a', b')\\
\; \;  \; \; \; \; =\; \; \; & 
 (b', a') \preccurlyeq (b, a)  \; \; \;  \text{if} \; \; \;
    (a,b) \preccurlyeq (a', b')\\
\; \; \; \; \; \; =\; \; \; & 
 b' \leq b\; \text{and}\; a' \geq a  \; \; \;  \text{if} \; \; \; 
    a \leq a' \; \text{and}\; b \geq b'
\end{align*}
Together, these properties define a De Morgan negation algebra \cite{Almeida09}, making both De Morgan equalities to hold, as well as $\neg \top = \bot$ and $\neg \bot = \top$. Similarly each pair of modalities is dual, as detailed in lemma~\ref{ldual}. Our syntax has, therefore, some redundancy. Note, however, that 
$$  
(a,b) \doublesqcap (\nneg (a,b)) \; =\; 
(a,b) \doublesqcap (b, a) \; =\;  
(a \sqcap b, b \sqcup a)
$$
which does not coincide with $(0,1)$, the interpretation of $\bot$, even if, in a sense, it `gets closer' by descending in the underlying lattice.

A negation operator $\simnot$ induced through implication, i.e.
$$(w\models \simnot \varphi)\; =\; (w\models \varphi \to \bot) \; =\; (w\models \varphi) \Longrightarrow (w\models \bot)$$
has a different behaviour. Actually, if the value of $\varphi$ in $w$ is $(a,b)$, the value of $\simnot \varphi$ becomes 
$(a \toh 0, a \sqcap 1) \, = \, (a \toh 0, a)$. This means that $\simnot \varphi$ in a state $w$ yields $\top$ iff the value of $\varphi$ in that state  has the form $(0,b)$, for any $b\in A$. This also means that the corresponding pre-minimality property depends on the specific underlying iMTL-algebra. The same applies to the duality of $\bot$ and $\top$. Both clearly hold for $\mathbf{2}$, $\mathbf{3}$ and $\G$, but not necessarily in general.  Clearly, this negation is not an involution, even if the De Morgan laws hold. Actually,

\begin{lemma} 
Let $\A$ be a iMTL-algebra. The De Morgan laws for $\simnot$ are valid in $\logica{\A}$.
\end{lemma}

\begin{proof}
Let us consider one of the laws; the proof of the other is similar.
\vspace{-0.5cm}
\begin{center}
 \begin{calculation} w \models \simnot (\varphi_1 \wedge \varphi_2)
\just = {definition of $\simnot$} 
w \models (\varphi_1 \wedge \varphi_2) \to \bot 
\just = {definition of $\models$} 
\left( w \models (\varphi_1 \wedge \varphi_2)\right) \Longrightarrow (w\models \bot) 
\just = {definition of $\models$} 
\left ( (w\models \varphi_1)\doublesqcap(w\models\varphi_2)\right)\Longrightarrow (0,1) 
\just = {making $a = (w\models \varphi_1)^+$, $b = (w\models \varphi_1)^-$, and similalry for $a'$, $b'$}
\left ( (a,b)\doublesqcap (a',b')\right) \Longrightarrow (0,1) 
\just = {definition of $\doublesqcap$} 
(a\sqcap a', b\sqcup b')\Longrightarrow (0,1)
\just = {definition of $\Longrightarrow$} 
((a\sqcap a')\rightharpoonup 0, a\sqcap a')
\just = {\eqref{eq4lm1}: $\bigsqcap _i a_i \rightharpoonup b = \bigsqcup_i(a_i\rightharpoonup b)$}
(a\to0 \sqcup a'\to0,a\sqcap a')
\just = {definition of $\doublesqcup$} 
(a\rightharpoonup 0,a)\doublesqcup(a'\rightharpoonup 0,a')
\just = {definition of $\Longrightarrow$; unfolding abbreviations $a,a',b,b'$}
\left( (w\models\varphi_1)\Longrightarrow(0,1)\right)\doublesqcup \left((w\models \varphi_2)\Longrightarrow(0,1)\right)
\just = {definition of $\simnot$} 
(w\models \simnot\varphi_1)\doublesqcup(w\models\simnot\varphi_2) 
\just = {definition of $\models$} 
w\models(\simnot \varphi_1 \vee \simnot \varphi_2)
\end{calculation}\qedhere
 \end{center}

\end{proof}
Can we directly compare both negations? Again the result depends on the underlying iMTL-algebra. For all the cases considered in this paper, however, a simple calculation leads to the following conclusion:
\begin{lemma} 
Let $\phi$ be a formula in $\logica{\A}$. For $\A$ instantiated to $\mathbf{2}$, $\mathbf{3}$ or $\G$,
\begin{align*}
(w\models \neg \phi) \, \preccurlyeq\, (w\models \simnot \phi)
& \; \text{ if }\; (w\models \phi)^+ = 0\\
(w\models \simnot \phi) \, \preccurlyeq\, (w\models \neg \phi)
& \; \text{ if }\; (w\models \phi)^+ \neq 0
\end{align*}
\end{lemma}

\begin{proof}
By direct calculation in each case.
\end{proof}

The two pairs of modalities behave as expected, quantifying over the positive and the negative component of the accessibility relation, as shown in the following lemma. This assigns an existential, `diamond-like' behaviour to 
$\slashed{\openbox}$, and dually a universal, `box-like' behaviour to
$\slashed{\Diamond}$. 

\begin{lemma}\label{ldual} Let $\A$ be an iMTL-algebra. The following equivalences hold in any $\A$-paraconsistent Kripke model:

\begin{minipage}[b]{0.45\linewidth}
 \begin{eqnarray}\label{ldual1} 
    \openbox \neg \varphi & \equiv & \neg \Diamond \varphi\\
\Diamond \neg\varphi & \equiv & \neg \openbox \varphi
\end{eqnarray}
\end{minipage}
\begin{minipage}[b]{0.45\linewidth}
\begin{eqnarray}
\slashed{\openbox}\neg\varphi & \equiv &\neg\slashed{\Diamond}\varphi \label{ldual3}\\
 \slashed{\Diamond}\neg\varphi & \equiv &\neg\slashed{\openbox}\varphi
\end{eqnarray}
\end{minipage}
\end{lemma}
\begin{proof}
We prove the 'box' cases, \eqref{ldual1} and \eqref{ldual3}; the remaining cases are similar.
 Hence, for \eqref{ldual3},
 \begin{calculation}
(w\models\slashed\openbox\neg\varphi)
\just = {definition of $\models$}
(\mydiamonminus (w,(\neg\varphi)^-), \boxminus (w, (\neg\varphi)^+) )
\just = {definition of $\boxminus$, $\mydiamonminus$}
\left(\bigsqcup\limits_{w'\in W} \{R^-(w,w')\sqcap (w'\models\neg \varphi)^-\},\bigsqcap\limits_{w'\in W} \{R^-(w,w')\rightharpoonup (w'\models\neg\varphi)^+\}\right)
\just = {definition of $\models$}
\left(\bigsqcup\limits_{w'\in W} \{R^-(w,w')\sqcap (\nneg(w'\models\varphi))^-\},\bigsqcap\limits_{w'\in W} \{R^-(w,w')\rightharpoonup (\nneg (w'\models\varphi))^+\}\right)
\just = {definition of $\nneg$}
\left(\bigsqcup\limits_{w'\in W} \{R^-(w,w')\sqcap (w'\models\varphi)^+\},\bigsqcap\limits_{w'\in W} \{R^-(w,w')\rightharpoonup (w'\models\varphi)^-\}\right)
\just = {definition of $\mydiamonminus$ and $\boxminus$}
(\mydiamonminus (w,\varphi^+),\boxminus (w,\varphi^-))
\just = {definition of $\nneg$}
\nneg(\boxminus(w,\varphi^-),\mydiamonminus (w,\varphi^{+}))
\just = {definition of $\models$}
(w\models\neg\slashed{\Diamond}\varphi)
\end{calculation}
The proof of \eqref{ldual1} has exactly the same structure, unfolding $(w\models\openbox\neg\varphi)$
into $(\boxplus (w, (\neg\varphi)^+),  \mydiamonplus (w,(\neg\varphi)^-))$, and then replacing $\boxplus$ and $\mydiamonplus$
by their definitions. The conclusion comes from swapping the formula through negation and rolling back the definitions.\qedhere

 \end{proof}

To conclude this section, let us look into the evaluation of a few, simple formulas, which may help to build up intuitions.  In the following three examples take $\G$ as the underlying iMTL-algebra.

\begin{example}
A very simple scenario illustrates the behaviour of both pairs of modalities in presence of `crisp' transitions, weighted by $(1,0)$ and $(0,1)$. Notice how the typical roles of the universal and existential modalities are reversed when going from the positive to the negative components of the accessibility relation.

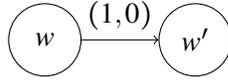
\begin{minipage}[b]{0.5\linewidth}
\begin{center}
\begin{tikzpicture}[->]
    \node[state] (1) at (0,0) {$w$};
    \node[state] (2) at (2,0) {$w'$};
    \draw (1) edge node[above] {$(1,0)$} (2);
\end{tikzpicture}
\end{center}
\vspace{0.35cm}
\begin{tabular}{c|c|c|c|c}
\hline \hline
    $V(p,w')$ & (0,1) & (1,0) &(0,0) &(1,1)  \\
\hline \hline
    $\openbox p$ & (0,1) & (1,0) & (0,0) & (1,1) \\ \hline
    $\Diamond p$ & (0,1) & (1,0) & (0,0) & (1,1) \\ \hline
    $\slashed{\openbox} p$ & (0,1) & (0,1) & (0,1) & (0,1) \\ \hline
    $\slashed{\Diamond} p$ & (1,0) & (1,0) & (1,0) & (1,0) \\ \hline \hline
\end{tabular}
\end{minipage}
 \begin{minipage}[b]{0.4\linewidth}
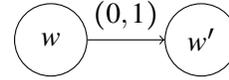
\begin{center}
    
\begin{tikzpicture}[->]
    \node[state] (1) at (0,0) {$w$};
    \node[state] (2) at (2,0) {$w'$};
    \draw (1) edge node[above] {$(0,1)$} (2);
\end{tikzpicture}
\end{center}
\vspace{0.35cm}
\begin{tabular}{c|c|c|c|c}
\hline \hline
   $V(p,w')$  & (0,1) & (1,0) &(0,0) &(1,1)  \\
    \hline \hline
    $\openbox p$ & (1,0) & (1,0) & (1,0) & (1,0) \\ \hline
    $\Diamond p$ & (0,1) & (0,1) & (0,1) & (0,1) \\ \hline
    $\slashed{\openbox} p$ & (1,0) & (0,1) & (0,0) & (1,1) \\ \hline
    $\slashed{\Diamond} p$ & (1,0) & (0,1) & (0,0) & (1,1) \\ \hline \hline
\end{tabular}
 \end{minipage}
\end{example}
 \vspace{0.5cm}
 
 \begin{example}
Still sticking to one-transition systems, observe now the effect of introducing arbitrary weights in both transitions and the valuation of propositions. The following table, sums up the results for different valuations in both states $w$ and $v$.

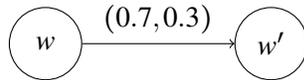
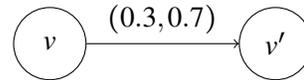
\vspace{0.35cm}
\begin{center}
\begin{tikzpicture}[->]
    \node[state] (1) at (0,0) {$w$};
    \node[state] (2) at (3,0) {$w'$};
    \node[state] (3) at (6,0) {$v$};
    \node[state] (4) at (9,0) {$v'$};
    \draw (1) edge node[above] {$(0.7,0.3)$} (2);
    \draw (3) edge node[above] {$(0.3,0.7)$} (4);
\end{tikzpicture}
\end{center}

\vspace{0.35cm}
\noindent\resizebox{\columnwidth}{!}{%
\begin{tabular}{c|c|c|c|c|c|c|c|c|c}
\hline \hline
    $V(p,x),  x\in\{w',v'\}$ & $(0,0)$ & $(0,0.5)$  & $(0,1)$ & $(0.5,1)$ & $(1,1)$ & $(1,0.5)$ & $(1,0)$ & $(0.5,0)$ & $(0.5,0.5)$\\
    \hline \hline
    $w\models\openbox p$ & $(0,0)$ & $(0,0.5)$ & $(0,0.7)$ & $(0.5,0.7)$ & $(1,0.7)$ & $(1,0.5)$ & $(1,0)$ & $(0.5,0)$ & $(0.5,0.5)$\\
    \hline
    $v\models\openbox p$ & $(0,0)$ & $(0,0.3)$ & $(0,0.3)$ & $(1,0.3)$ & $(1,0.3)$ & $(1,0.3)$ & $(1,0)$ & $(1,0)$ & $(1,0.3)$\\
    \hline
    $w\models\Diamond p$ & $(0,0)$ & $(0,0.5)$ & $(0,1)$ & $(0.5,1)$ & $(0.7,1)$ & $(0.7,0.5)$ & $(0.7,0)$ & $(0.5,0)$ & $(0.5,0.5)$\\
    \hline
    $v\models\Diamond p$ & $(0,0)$ & $(0,1)$ & $(0,1)$ & $(0.3,1)$ & $(0.3,1)$ & $(0.3,1)$ & $(0.3,0)$ & $(0.3,0)$ & $(0.3,1)$\\
    \hline
    $w\models\slashed{\openbox} p$ & $(0,0)$ & $(0.3,0)$ & $(0.3,0)$ & $(0.3,1)$ & $(0.3,1)$ & $(0.3,1)$ & $(0,1)$ & $(0,1)$ & $(0.3,1)$\\
    \hline
    $v\models\slashed{\openbox} p$ & $(0,0)$ & $(0.5,0)$ & $(0.7,0)$ & $(0.7,0.5)$ & $(0.7,1)$ & $(0.5,1)$ & $(0,1)$ & $(0,0.5)$ & $(0.5,0.5)$\\
    \hline
    $w\models\slashed{\Diamond} p$ & $(0,0)$ & $(1,0)$ & $(1,0)$ & $(1,0.3)$ & $(1,0.3)$ & $(1,0.3)$ & $(0,0.3)$ & $(0,0.3)$ & $(1,0.3)$\\
    \hline
    $v\models\slashed{\Diamond} p$ & $(0,0)$ & $(0.5,0)$ & $(1,0)$ & $(1,0.5)$ & $(1,0.7)$ & $(0.5,0.7)$ & $(0,0.7)$ & $(0,0.5)$ & $(0.5,0.5)$\\
    \hline \hline
\end{tabular}%
}

\vspace{0.35cm}
\noindent
With no surprise, we conclude that 
$w\models\openbox p = v\models \neg\slashed{\openbox}p$,
$w\models \Diamond p = v\models\neg\slashed{\Diamond} p$,
$w\models \slashed{\openbox} p = v\models\neg\openbox p$, and
$w\models \slashed{\Diamond} p = v\models\neg\Diamond p$.
 \end{example}
 \vspace{0.5cm}
 
\begin{example}
Extending the model with an additional node as in
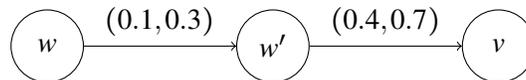
\begin{center}
\begin{tikzpicture}[->]
    \node[state] (1) at (0,0) {$w$};
    \node[state] (2) at (3,0) {$w'$};
    \node[state] (3) at (6,0) {$v$};
    \draw (1) edge node[above] {$(0.1,0.3)$} (2);
    \draw (2) edge node[above] {$(0.4,0.7)$} (3);
\end{tikzpicture}
\end{center}
we may evaluate deeper formulas. For example, 
$V(p,w')=V(p,v)=(1,0)$, entails 
$$
\openbox \openbox p \; =\; (1,0),\; \;  \; \;  \; \Diamond \openbox p\; =\;   (0.1,0), 
\; \;  \; \;  \; \slashed{\openbox}\openbox p\; =\; (0,1), \; \;  \; \;  \;\slashed{\Diamond} \openbox p\; = \;  (0,1). 
$$
and
$$
\slashed{\Diamond} \Diamond  p\; =\; (0,0.3), \; \;  \; \;  \;
\Diamond \slashed{\Box}  p \; =\;  (0.1,0), \; \;  \; \;  \;
\slashed{\Diamond}\slashed{\Box} p \; =\;  (0,1),\; \;  \; \;  \;
\slashed{\Diamond}\slashed{\Diamond} p\; =\;  (0,0.1).
$$
\end{example}
 \vspace{0.5cm}
 
\begin{example} As a last example, consider $\mathbf{3}$ as the underlying iMTL-algebra, and the following model

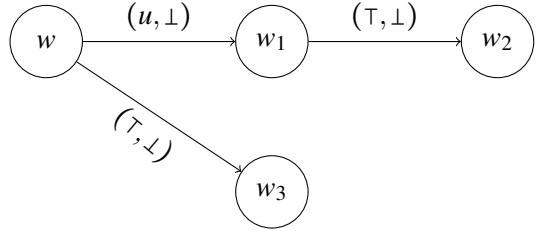
\begin{center}
\begin{tikzpicture}[->]
    \node[state] (1) at (0,0) {$w$};
    \node[state] (2) at (3,0) {$w_1$};
    \node[state] (3) at (6,0) {$w_2$};
    \node[state] (4) at (3,-2) {$w_3$};
    \draw (1) edge node[above] {$(u,\bot)$} (2);
    \draw (2) edge node[above] {$(\top,\bot)$} (3);
    \draw (1) edge node[below, rotate = -40] {$(\top,\bot)$} (4);
\end{tikzpicture}
\end{center}

A valuation such that $V(p,w_1)=V(p,w_2)=(\top,\bot)$ and $V(p,w_3)=(u,u)$ entails

\vspace{0.35cm}

\begin{minipage}[b]{0.23\linewidth}
\begin{tabular}{c|c}
\hline \hline
     $w\models \openbox p$ & $(u,u)$   \\
     \hline
     $w\models \openbox \openbox p$ & $(\top,\bot)$\\
     \hline
     $w\models \Diamond \openbox p$ & $(\top,\bot)$ \\
     \hline
     $w\models \slashed{\openbox}\openbox p$ & $(\bot,\top)$ \\
     \hline
     $w\models \slashed{\Diamond} \openbox p$ & $(\bot,u)$\\
     \hline \hline
\end{tabular}
\end{minipage}
\begin{minipage}[b]{0.23\linewidth}
\begin{tabular}{c|c}
\hline \hline
     $w\models \Diamond p$ & $(u,\bot)$   \\
     \hline
     $w\models \openbox \Diamond p$ & $(\bot,\top)$\\
     \hline
     $w\models \Diamond \Diamond p$ & $(u,\bot)$ \\
     \hline
     $w\models \slashed{\openbox}\Diamond p$ & $(u,\bot)$ \\
     \hline
     $w\models \slashed{\Diamond} \Diamond p$ & $(\top,\bot)$\\
     \hline \hline
\end{tabular}
\end{minipage}
\begin{minipage}[b]{0.23\linewidth}
\begin{tabular}{c|c}
\hline \hline
     $w\models \slashed{\openbox} p$ & $(\bot,\top)$   \\
     \hline
     $w\models \openbox \slashed{\openbox} p$ & $(u,\bot)$\\
     \hline
     $w\models \Diamond \slashed{\openbox} p$ & $(u,\bot)$ \\
     \hline
     $w\models \slashed{\openbox}\slashed{\openbox} p$ & $(\bot,\top)$ \\
     \hline
     $w\models \slashed{\Diamond} \slashed{\openbox} p$ & $(\bot,u)$\\
     \hline \hline
\end{tabular}
\end{minipage}
\begin{minipage}[b]{0.23\linewidth}
\begin{tabular}{c|c}
\hline \hline
     $w\models \slashed{\openbox} p$ & $(\bot,u)$   \\
     \hline
     $w\models \openbox \slashed{\openbox} p$ & $(\bot,u)$\\
     \hline
     $w\models \Diamond \slashed{\openbox} p$ & $(\bot,u)$ \\
     \hline
     $w\models \slashed{\openbox}\slashed{\openbox} p$ & $(u,\bot)$ \\
     \hline
     $w\models \slashed{\Diamond} \slashed{\openbox} p$ & $(\top,\bot)$\\
     \hline \hline
\end{tabular}
\end{minipage}

\end{example}

\section{Bisimilarity and invariance}
\label{secm}

This section characterises simulation and bisimulation for $\A$-paraconsistent models.

\begin{definition}\label{simulation} Let $\A$ be an iMTL-algebra enriched with a metric. Let $T_1=(W_1,R_1)$, $T_2=(W_2, R_2 )$ be two $\A$-PTS.

A relation 
 $S \subseteq W_1 \times W_2$ is a simulation if, for any $(w_1,w_2)\in S$, 
\[w_1  \xrightarrow{(a,b)\text{ }}_{T_1}w_1' \Rightarrow 
\langle \exists w_2'\in W_2, \exists (a',b')\in A \times A:w_2  \xrightarrow{(a',b')\text{ }}_{T_2}w_2' \wedge  (w_1',w_2')\in S \text{ and } (a,b)\preccurlyeq (a',b')\rangle\]
\end{definition}

\begin{example}\label{simeg} Let us consider the $\G$-PTS of Example~\ref{exframes}.
Observe that  $w_1 \lesssim v_1$, because there is a simulation, $S=\{ (w_1,v_1),( w_2,v_2),( w_3,v_2),(w_4,v_3),( w_5,v_4)\}$ that contains $(w_1, v_1)$.

\end{example}

\begin{definition} Let $\A$ be an iMTL-algebra enriched with a metric. 
Let $M_1=(W_1,R_1,V_1)$ and $M_2=(W_2,R_2,V_2)$ two $\A$-paraconsistent Kripke models. The relation $S\subseteq W_1\times W_2$ is a \emph{simulation between $M_1$ and $M_2$} if 
\begin{itemize}
\item $S$ is a simulation between the $\A$-PTS $(W_1,R_1)$ and $(W_2,R_2)$
\item for any $p\in\Prop$ and $(w,v) \in S$, $V_1(w,p) \preccurlyeq V_2(v,p)$
\end{itemize}
\end{definition}

\begin{lemma}\label{lm:simprev} Let $\A$ be an iMTL-algebra enriched with a metric.  If $S\subseteq W_1\times W_2$ is a simulation between $\A$-paraconsistent Kripke models $M_1$ and $M_2$ and $\langle w_1,w_2\rangle\in S$ then 
$(w_1\models \varphi)\preccurlyeq (w_2\models\varphi) \text{ for }\varphi \in \Fm^{+\Diamond}$
where $\text{Fm}^{+\Diamond}$ is the positive fragment of $\logica{\A}$ with the modal connective $\Diamond$, but omitting the consistency connective $\circ$.
\end{lemma}
\begin{proof} 
The proof proceeds by induction over the structure of formulas.
\begin{description}


\item[Case $\varphi=\bot$:]
 We have $(w\models\bot)=(1,0)=(v\models\bot )$. Hence $(w\models\bot)\preccurlyeq(v\models\bot)$.
\item [Case $\varphi=p$, $ p \in Prop$:]
Since $S$ is simulation, $V_1^+(w,p)\leq V_2^+(v,p)$  and $V_1^-(w,p)\geq V_2^-(v,p)$. Hence,
$$(w\models p)=(V_1^+(w,p),V_1^-(w,p)) \;  \preccurlyeq \; (V_2^+(w,p),V_2^-(w,p)) = (v\models p)$$.
\item[Case $\varphi=\varphi_1\wedge\varphi_2$:]
By induction hypothesis and monotonicity,
$$(w\models\varphi_1)\doublesqcap
(w\models\varphi_2) \preccurlyeq (v\models\varphi_1)\doublesqcap (v\models\varphi_2)$$
\noindent
Hence, $(w\models \varphi_1\wedge\varphi_2)\preccurlyeq(v\models \varphi_1 \wedge\varphi_2)$.




\item[Case $\varphi= \varphi_1\vee\varphi_2$:] Analogous to the previous case.
\item[Case $\varphi = \Diamond \varphi_1$:]
We have to prove that 
\begin{align*} 
& \big(\Max_{w'\in R_1[w]} \{R_1^+(w,w')\sqcap(w'\models\varphi_1)^+\},\Min_{w'\in R_1[w]} \{R_1^+(w,w')\rightharpoonup(w'\models\varphi_1)^-\}\big)\\
\preccurlyeq \; \; \; & \;\big(\Max_{v'\in R_1[v]} \{R_1^+(v,v')\sqcap(v'\models\varphi_1)^+\},\Min_{v'\in R_1[v]} \{R_1^+(v,v')\rightharpoonup(v'\models\varphi_1)^-\}\big)
\end{align*}
 Observe that, since $(w,v)\in S$, we have by induction hypothesis that, for any $w'\in R_1[w]$ with $R_1^+(w,w')\sqcap(w'\models\varphi_1)^+$, there is a $v'\in R_1[v]$ such that $R_1^+(w,w')\sqcap(w'\models\varphi_1)^+\; \; \leq \; \; R_1^+(v,v')\sqcap(v'\models\varphi_1)^-$. Hence, by monotonicity of $\sqcup$, 
$$\Max_{w'\in R_1[w]} \{R_1^+(w,w')\sqcap(w'\models\varphi_1)^+\; \; \leq \; \; \Max_{w'\in R_1[w]} \{R_1^+(w,w')\sqcap(w'\models\varphi_1)^+\}$$
\noindent
Analogously, one may prove that
\begin{displaymath}
\Min_{w'\in R_1[w]} \{R_1^+(w,w')\rightharpoonup(w'\models\varphi_1)^-\} \; \; \geq \; \; \Min_{v'\in R_1[v]} \{R_1^+(v,v')\rightharpoonup(v'\models\varphi_1)^-\}.\qedhere\end{displaymath}
\end{description}
\end{proof}

\noindent
Note there is a number of formulas which are not preserved by these mappings. Such are the cases of $\neg\varphi$, $\varphi_1\to\varphi_2$, $\openbox\varphi$, $\slashed{\openbox}\varphi$ and $\slashed{\Diamond}\varphi$. A systematic study of these cases, with counterexamples, can be found in \cite{Ana21}.

\begin{definition}
\label{bisimulation2} Let $\A$ be an iMTL-algebra enriched with a metric. 
Let $T_1=\langle W_1,R_1\rangle$ and $T_2=\langle W_2,R_2\rangle$ be two $\A$-PTS. A relation $B\subseteq W_1\times W_2$ is a bisimulation if for $\langle p,q\rangle\in B$:



$p\xrightarrow{(a,b)\text{ }}_{M_1}p' \Rightarrow \langle \exists q'\in W_2:q  \xrightarrow{(a,b)\text{ }}_{M_2}q' \text{ and } \langle p',q'\rangle\in B\rangle$ and

$q\xrightarrow{(a,b)\text{ }}_{M_2}q' \Rightarrow \langle \exists p'\in W_1:p  \xrightarrow{(a,b)\text{ }}_{M_1}p' \text{ and } \langle p',q'\rangle\in B\rangle$.
\end{definition}

\noindent
Two states $p$ and $q$ are \emph{bisimilar}, written $p\sim q$, if there is a bisimulation $B$ such that $\langle p,q\rangle \in B $.


\begin{definition} A relation $B\subseteq W_1\times W_2$ is a bisimulation between two $\A$-paraconsistent Kripke models $M_1=(W_1,R_1,V_1)$ and $M_2=(W_2,R_2,V_2)$ if 
\begin{itemize}
\item  $B$ is a bisimulation between the $\A$-PTS $(W_1,R_1)$ and $(W_2,R_2)$
\item for any $p\in\Prop$ and $(w_1,w_2) \in B$, $V_1(w_1,p) = V_2(w_2,p)$.
\end{itemize}
\end{definition}

\begin{lemma} Let 
$M_1=(W_1,R_1,V_1)$ and $M_2=(W_2,R_2,V_2)$ two $\A$-paraconsistent Kripke models for $\Prop$ and 
$B\subseteq W_1 \times W_2$ a bisimulation. Then, for any $(w,v)\in B$ and for any $\varphi \in \Fm(\Prop)$, 
$(w\models\varphi)=(v\models\varphi)$.
\end{lemma}
\begin{proof} 
The proof proceeds by induction on the structure of the formulas.
\begin{itemize}
\item \textbf{Case $\varphi=\bot$}:
by definition of $\models$, $(w\models\bot)=(0,1)=(v\models\bot)$.

\item \textbf{Case $\varphi=p$, $p\in Prop$}:
since by hypothesis $(w,v)\in B$, we have that $V_1(w,p)=V_1(v,p)$. Hence 
$(w\models p)=V_1(w,p)=V_1(v,p)=(v\models p)$.

\item \textbf{Case $\varphi=\neg\varphi_1$}:
$(w\models\neg\varphi_1)=((w\models\varphi_1)^-,(v\models\varphi_1)^+)=(v\models\neg\varphi_1)$, and, by induction hypothesis,
$(w\models\varphi_1)^+)= ((v\models\varphi_1)^-$.

\item \textbf{Case $\varphi=\varphi_1\wedge\varphi_2$}: the proof is analogous to the case in Lemma~\ref{lm:simprev}.

\item  \textbf{Cases} $\varphi=\varphi_1\vee\varphi_2$ and $\varphi=\varphi_1\to\varphi_2$: analogous to the previous case.

\item  \textbf{Case $\varphi=\Diamond\varphi_1$}: analogous to the proof of the case in Lemma~\ref{lm:simprev}.

\item \textbf{ Case $\varphi=\openbox\varphi_1$}: analogous to the previous case.

\item\textbf{Case $\varphi =\slashed{\openbox}\varphi_1$}:
Observe that 
$$(w\models\slashed{\openbox}\varphi_1)\; \; =\; \; (\Max_{w'\in W_1} (R^-(w,w')\sqcap(w'\models\varphi_1)^-),\Min_{w'\in W_1} (R^-(w,w')\rightharpoonup(w'\models\varphi_1)^+))$$
By induction hypothesis, this is equal to 
$$(\Max_{v\in W_2} (R^-(v,v')\sqcap(v'\models\varphi_1)^-),\Min_{v\in W_2} (R^-(v,v')\rightharpoonup(v'\models\varphi_1)^+)) \; \; =\; \; (v\models\slashed{\openbox} \varphi_1)$$

\item \textbf{Case $\varphi =\slashed{\Diamond}\varphi_1$}: is analogous to the previous case.
\item \textbf{Case $\varphi =\circ\varphi_1$}:
since $(w\models\varphi)=(v\models\varphi)$ it is straightforward that $(w\models\circ\varphi)=(v\models\circ\varphi)$
\qedhere
\end{itemize}
\end{proof}

\section{Conclusions and future work} 


We have discussed a method to ``build paraconsistent modal logics on-demand" (in style of \cite{MADEIRA20161011}), i.e. parametric in an  iMTL-algebra equipped with a metric space over its carrier. As explained, these results generalise the modal extension of Belnap-Dunn four-valued logic \cite{10.1093} to deal with more complex truth spaces. An earlier work by V. Goranko \cite{Gor90}, which 
discusses axiomatizations and completeness of bimodal logics with modalities associated with an accessibility relation and its complement,
bears an interesting correspondence with the structures proposed in this paper. In the particular, the case in which the weights labelling the positive and the negative components of the accessibility relation are complementary is discussed there.

A lot remains to be done. The main challenge is the development of a (also parametric)  proof calculus for $\logica{\A}$, and study its properties.
From the applications side, on the other hand, we intend to pursue its application to optimization of quantum circuits \cite{Ana21}. A realted work direction focus on the development of an algebra of $\A$-paraconsistent transition systems, to proceed towards more complex systmes in a compositional way. This will also open possibility of  developing  a dynamic logic for weighted  (e.g. probabilisitc, quantum) programs on top of this semantic framework and its algebra. We believe this may be relevant, namely for verification of quantum programs.


    

\paragraph{Acknowledgements.}
The authors are grateful to V. Goranko for calling attention to his earlier work \cite{Gor90} and helpful insights. 
The comments from the anonymous referees were also greatly appreciated. 

This work was funded in the context of project IBEX (\texttt{PTDC/CCI-COM/4280/2021}), by \textsc{Fct}, the Portuguese funding agency for Science and Technology.

\bibliographystyle{eptcs}
\bibliography{biblio}

\end{document}